\documentclass[12pt]{article}
\usepackage{amssymb,latexsym}
\usepackage[cmex10]{amsmath}
\newtheorem{theorem}{Theorem}[section] 
\newtheorem{corollary}[theorem]{\bf Corollary} 
\newtheorem{proposition}[theorem]{\bf Proposition} \newtheorem{definition}[theorem]{\bf Definition} \newtheorem{lemma}[theorem]{\bf Lemma} \newtheorem{example}[theorem]{\bf Example}  
\newtheorem{algorithm}[theorem]{\bf Algorithm} 
\newtheorem{remark}[theorem]{\bf  Remark} \newtheorem{remarks}[theorem]{\bf  Remarks} \newenvironment{proof}{{\bf Proof.}}{\hspace*{\fill}$\blacksquare$\par\vspace{4mm}} 
\newcommand{\Ann} {\mathrm{Ann}}
\newcommand{\Min} {\mathrm{MP}}
\newcommand{\Rp} {\mathrm{RP}}
\newcommand{\SRp} {\mathrm{SRP}}
\newcommand{\J } {\mathrm{J}}

\newcommand{\ol} {\overline}

\newcommand{\ra} {\rightarrow}
\newcommand{\ul} {\underline}
\newcommand{\ee} {\mathrm{e}}
\newcommand{\pp} {\mathrm{p}}

\newcommand{\LC} {\mathrm{L}}

\newcommand{\F}{\mathbb{F}}
\newcommand{\N}{\mathbb{N}}
\newcommand{\Z}{\mathbb{Z}}
\title{The Berlekamp-Massey Algorithm via Minimal Polynomials}
\author{G. H. Norton, Department of Mathematics\\ University of Queensland.}
\begin{document}
\maketitle
\begin{abstract}
We present a recursive minimal polynomial theorem for  finite sequences over a commutative integral domain $D$. This theorem  is relative to any element of $D$. The ingredients are: the arithmetic of Laurent polynomials over $D$, a recursive 'index function' and simple mathematical induction. Taking reciprocals gives a 'Berlekamp-Massey theorem'  i.e. a recursive construction of the polynomials arising in the Berlekamp-Massey algorithm, relative to any element of $D$.  The recursive theorem readily yields the  iterative minimal polynomial algorithm due to the author and a transparent derivation of the iterative Berlekamp-Massey algorithm.

We give an upper bound for the sum of the linear complexities of $s$ which is tight if $s$ has a perfect linear complexity profile. This implies that over a field, both iterative algorithms require at most $2\lfloor \frac{n^2}{4}\rfloor$ multiplications. 
\end{abstract}


\section{Introduction}
\subsection{The Berlekamp-Massey (BM) Algorithm}
The BM algorithm determines a  linear recurrence of least order $\LC\geq 0$ which generates a given (finite) sequence $s$ of length $n\geq 1$ over a field $\F$, \cite{Ma69}. It is widely used in Coding Theory, Cryptography  and Symbolic Computation. There are also connections with partial realization in Mathematical Systems Theory; see \cite[Introduction]{N95b} and the references cited there.  For an exposition based on \cite{Ma69}, see \cite[Section 9.5]{PW}. 

However, 'The inner workings of the Berlekamp-Massey algorithm can appear somewhat mysterious', \cite[p. 187]{BL83} and the extended Euclidean algorithm is usually preferred as 'It is much easier to understand', \cite[p. 355]{McE02}. For a recent example where the extended Euclidean algorithm is regarded as 'simpler to understand, to implement and to prove', see \cite{ABN}.

There have been a number of derivations of the BM algorithm for sequences over a field, such as \cite[Chapter 7]{BL83} and  \cite{IY}, which uses Hankel matrices. A similar approach to  \cite{IY} appeared in  \cite{JM}; this uses Kronecker's Theorem on the rank of Hankel matrices and the Iovidov index of a Hankel matrix. We do not know if \cite{JM} applies to finite fields. Another approach \cite{HJ}  uses the Feng-Tzeng algorithm \cite{FT91}. For references relating the BM and Euclidean algorithms, see \cite[Introduction]{N95b} and the references cited there.  A recursive version of the BM algorithm (based on splitting a sequence and recombining the results) appeared in \cite[p. 336]{BL83}. 
\subsection{Linear Recurring Sequences via Laurent Series}
The conventional approach to linear recurring sequences   indexes them by the non-negative integers and uses reciprocals of polynomials as characteristic functions; see  \cite{LN83}, \cite{Ma69} or any of the standard texts. This complicates their theory. 

We took a non-standard, algebraic approach in \cite{N95b}, \cite{N99b} (an expository version of \cite{N95b}): use the field $\F[[x^{-1},x]$ of $\F$-Laurent series in $x^{-1}$ (the case $\F=\mathbb{R}$ is widely used in Mathematical Systems Theory) to study linear recurring sequences. 
For us, a sequence is indexed by $1,2,\ldots$ We began with $\F[[x^{-1},x]$  as standard $\F[[x^{-1},x]$-module. Later we realized that it was enough for $\F$ to be a  commutative unital integral domain $D$ and used the decomposition $$D[[x^{-1},x]=x^{-1}D[[x^{-1}]]\oplus D[x].$$ 

The action of $D[x]$ on $x^{-1}D[[x^{-1}]]$ is obtained as follows: project multiplication (in $D[[x^{-1},x]$) of an element of $D[x]$ and an element of $x^{-1}D[[x^{-1}]]$ onto the first summand. One checks that this makes $x^{-1}D[[x^{-1}]]$ into a $D[x]$-module. Linear recurring sequences are then the torsion elements in a natural $D[x]$-module. For any sequence $s$, we have its  annihilator ideal $\Ann(s)$; it elements are the 'annihilating polynomials' of $s$ and are defined by Equation (\ref{lrs}). Strictly speaking, $s$ satisfies a linear recurrence relation if $\Ann(s)\neq \{0\}$ and is a linear recurring sequence if it has a monic annihilating polynomial. 

When $D$ is a field, $\Ann(s)$ is generated by a unique monic annihilating polynomial of $s$, the minimal polynomial of $s$ (rather than the conventional reciprocal of a certain characteristic polynomial multiplied by a power of $x$).   
In \cite[Section IIA ]{S05}\footnote{In \cite{S05}, an element of $\Ann(s)$ with minimal degree was called 'a characteristic polynomial' of $s$.}, \cite[Definition 2.1]{AS} and \cite[Definition 2.1]{Salagean}, the definition of a linear recurring sequence $s_0,s_1,\ldots$ is equivalent to  expanding the left-hand side of Equation (\ref{lrs}) and replacing $d+1\leq j$ by $d\leq j$. We note that  \cite{N95b} and \cite{N99b} were referred to in \cite{NS-key}.
\subsection{Finite Sequences via Laurent polynomials}
To study {\em finite} sequences, we replaced Laurent series in $x^{-1}$ by {\em Laurent polynomials} $D[x^{-1},x]$  in  \cite{N95b}, \cite{N99b}; for a succinct overview of  \cite{N99b}, see \cite{N09d}. 
Unfortunately,
$x^{-1}D[x^{-1}]$ does not become a $D[x]$-module, but we can still define the notions of annihilating  and minimal polynomials; see Definitions \ref{anndefn}, \ref{mindefn}.  

In this paper, we present a recursive minimal polynomial function, see Section \ref{mptheory}. We replace the key definition of '$m$' of \cite[Equation (11), p. 123]{Ma69}  by a recursively defined 'index function'; see Definition \ref{indices}. We then derive a recursive theorem for minimal polynomials. Taking reciprocals (see Corollary \ref{recip}) leads to a recursive BM theorem (see Theorem \ref{one-step BM}). Our proofs use no more than the absence of zero-divisors, the arithmetic of Laurent polynomials and simple induction. 
\subsection{The Iterative Algorithms}
Our iterative minimal polynomial algorithm (Algorithm \ref{rewrite}) and   version of the BM algorithm (Algorithm \ref{newBMa}) follow immediately.  Both are relative to any scalar $\varepsilon\in D$ ($\varepsilon=1$ was used in \cite{Ma69} whereas $\varepsilon=0$ was used in \cite{N95b}, \cite{N99b}). Algorithm \ref{newBMa} is simpler than \cite[p. 148]{IY} --- see  Remark \ref{IY} --- and unlike the classical BM algorithm,  it is division-free, cf. \cite{divfree}.

The last section discusses the complexity of these two algorithms and does not depend on any aspects of the classical BM algorithm.  We give an upper bound for the sum of the linear complexities of $s$, which is tight if $s$ has a perfect linear complexity profile, Corollary \ref{sum.l}. This implies that the number of multiplications for Algorithms \ref{rewrite} and  \ref{newBMa} is at most $3\lfloor{n^2}/{4}\rfloor$ (Theorem \ref{mult}) and improves the bound of $\lfloor{3n^2}/{2}\rfloor$ given in \cite[Proposition 3.23]{N95b}.
Over a field $\F$, this reduces to $2\lfloor{n^2}/{4}\rfloor$ (if we ignore divisions in $\F$). We also include some remarks on the average complexity.
\subsection{Extensions and Rational Approximation}
Let $s=(s_1,\ldots,s_n)\in D^n$ be a finite sequence and $s^{(j)}=(s_1,\ldots,s_j)$ have 'generating function' $\ul{s^{(j)}}=s_1x^{-1}+\ldots +s_jx^{-j}$ for $1\leq j\leq n$. We write 

(i) $\mu^{(j)}$ for the minimal polynomial of $s^{(j)}$ of Theorem \ref{bit} with degree $\LC_j$

(ii) $\nu^{(j)}$ for the 'polynomial part'  of $\mu^{(j)}\cdot \ul{s}$,   which was evaluated in \cite{N95b}.

\noindent Then $\deg (\nu^{(j)})< \LC_j$ and
\begin{equation}\label{nums}
\mu^{(j)}\cdot \ul{s}\equiv 0\bmod x^{-j-1}
\end{equation}
for $1\leq j\leq n$. Remarkably, our formula for $\nu^{(j)}$ 
$$\nu^{(j)}=
\Delta'_j\cdot x^{\max\{e,0\}}\ \nu^{(j-1)}-\Delta_{j}\cdot x^{\max\{-e,0\}}\ \nu'^{(j-1)}
$$
is identical to Theorem (\ref{bit}) with $\mu$ replaced by $\nu$, where $e=\ee_{j-1}=j-2\LC_{j-1}$. The only difference being that $\nu$ is initialised differently. 

It is well-known that the BM algorithm also computes rational approximations.
We could also extend Algorithm \ref{newBMa} to compute $\nu^{(j)\ast}$ iteratively, obtaining $\deg(\nu^{(j)})$ from Equation (\ref{nums}) and $\LC_j$ (when $\deg(\nu^{(j)})\neq 0$).  In this way, Algorithm \ref{newBMa} could also be used to decode not just binary BCH codes, but Reed-Solomon codes,  errors and erasures,  classical Goppa codes, negacyclic codes and can be simplified in characteristic two. As this has already been done more simply using rational approximation via minimal polynomials in \cite{N95c} and \cite[Section 8]{N99b}, we will not compute $\nu^{(n)\ast}$ iteratively here. 
An extension of Theorem \ref{bit} to rational approximation will appear in \cite{N10b}.\\

We thank an anonymous referee for a simpler proof  of Lemma \ref{sum.l}.  A preliminary version of this work was presented in May 2010 at Equipe SECRET,
Centre de Recherche, INRIA Paris-Rocquencourt, whom the author  thanks for their hospitality.
\section{Preliminaries}
\subsection{Notation}
Let $\N=\{1,2,\ldots\}$, $n\in\N$ and let $D$ denote a commutative, unital integral domain with $1\neq 0$. For any set $S$ containing 0, $S^\times=S\setminus\{0\}$. We say that $f\in D[x]^\times$ is monic if its leading term is 1. The reciprocal of 0 is 0 and for $f\in D[x]^\times$, its reciprocal is $f^\ast(x)=x^{\deg(f)}f(x^{-1})$.  We  often write  $f=x^eg+h$ for $f(x)=x^eg(x)+h(x)$, where $e\in\N$ and $g,h\in D[x]$.

\subsection{Linear Recurring Sequences}
By an infinite sequence $s=(s_1,s_2,\ldots)$ over $D$, we mean a function  $s:\N\rightarrow D$ i.e. an element  of the abelian group $D^\N$. The standard algebraic approach to 'linear recurring sequences' is to study $D^\N$ using $\ol{s}(x)=\sum_{j\geq 1}s_jx^{j}\in D[[x]]$ as in \cite{LN83}, \cite{R86}, which requires reciprocal polynomials and complicates their underlying theory.

We recall the approach of \cite{N95b}. We begin with the standard  $D[[x^{-1}]$-module i.e. acting on itself via multiplication. (This also makes $D[[x^{-1}]$ as a $D[x]$-module.)  Next we let $D[x]$ act on $x^{-1}D[[x^{-1}]]$ by projecting the product  $f\in D[x]$ and $\ul{s}=\sum_{j\geq 1}s_jx^{-j}$ onto the first summand of $D[[x^{-1}]=x^{-1}D[[x^{-1}]\oplus D[x]]$ i.e.
$$f\circ \ul{s}=\sum_{j\geq 1}(f\cdot \ul{s})_{-j}\ x^{-j}.$$
 One checks that this makes $x^{-1}D[[x^{-1}]]$ into a $D[x]$-module.  Let 
 $$\mathrm{Ann}(\ul{s})=\{f\in D[x]:\ f\circ \ul{s}=0\}$$
denote the  {\em annihilator ideal} of $\ul{s}$; $f$ is an {\em annihilating polynomial or an annihilator} of $s$ if $f\in \Ann(\ul{s})$. We will often write $f\circ s$ for $f\circ \ul{s}$ and $\Ann(s)$  for $\Ann(\ul{s})$.  

We say that $s$ satisfies a {\em linear recurrence relation}  if it is a torsion element i.e. if $\mathrm{Ann}({s})\neq (0)$ \cite[Section 2]{N95b},  \cite[Section 2]{N95b}. In other words, $s$ satisfies a { linear recurrence relation}   if for some $f\in D[x]$ with $d=\deg(f)\geq 0$\begin{equation}  \label{lrs}
(f\cdot \ul{s})_{d-j}=0\mbox{ for } d+1\leq j.
\end{equation}
In this case,  $f\in\Ann(s)^\times$.
If we expand the left-hand side of Equation (\ref{lrs}) we obtain
$$f_0s_{j-d}+\cdots+f_ds_ {j}=0\mbox{ for }d+1\leq j.$$ When $f_d=1$, we can write 
$s_{j}=-(f_0s_{ j-d}+\cdots+ f_{ d-1}s_{ j-1})$ for $j\geq d+1$ and $s$ is a {\em linear recurring sequence}. For  the Fibonacci sequence $s=1,1,2,\ldots$ for example, $x^2-x-1\in\Ann(s)$.

We say that $f\in\Ann(s)^\times$ is a {\em minimal polynomial of $s$} if $$\deg(f)=\min\{\deg(g):\ g\in\Ann(s)\}.
$$
As $\Ann(s)$ is an ideal, we easily see that $s$ has a unique monic minimal polynomial which generates $\Ann(s)$ when $D$ is a field. More generally, it was shown in \cite{FN95} that if $\Ann(s)\neq \{0\}$ then\\ 

(i) if $D$ a factorial then $\Ann(s)$ is principal and has a primitive generator 

(ii) if  $D$ is potential, then $\Ann(s)$ has a unique monic generator. \\

\noindent In \cite{FN95},  we called $D$ potential if $D[[x]]$ is factorial. It is known that principal ideal domains and $\F[x_1,\ldots,x_k]$ are potential, but not all factorial domains are potential; see \cite[Introduction]{FN95} and the references cited there.
\subsection{Finite Sequences}
We now adapt the preceding definition of $\Ann(s)$  to finite sequences  $s\in D^n$ by using Laurent {\em polynomials}. This also leads to a less complicated theory of their annihilating and minimal polynomials.  

First, let $s=(s_1,\ldots,s_n)$ and $\overline{s}\in D[x]$ be $\overline{s}(x)=s_1x+\cdots+s_nx^n.$ 
We will also abbreviate $\overline{s}(x^{-1})=s_1x^{-1}+\cdots+s_nx^{-n}$ to $\ul{s}$, so that $\ul{s}_j=s_{\ -j}$ for $-n\leq j\leq -1$. 
In the following definition, multiplication of $f\in D[x]$ and $\ul{s}\in D[x^{-1}]$ is in the domain of $D$-Laurent {\em polynomials} $D[x,x^{-1}]$.
\begin{definition}[Annihilator, annihilating polynomial] \label{anndefn}(\cite[Definition 2.7, Proposition 2.8]{N95b}) If $s\in D^n$, then $f\in D[x]$  is an annihilator (or a characteristic polynomial) of $s$ if $f=0$ or $d=\deg(f)\geq 0$ and
\begin{equation}  \label{fseq}
(f\cdot \ul{s})_{d-j}=0\mbox{ for } d+1\leq j\leq n
\end{equation}
written $f\in \Ann(s)$. 
\end{definition}
If we expand the left-hand side of Equation (\ref{fseq}), we obtain
$$f_0s_{j-d}+\cdots+f_ds_ {j}=0\mbox{ for } d+1\leq j\leq n.$$
Any polynomial of degree at least $n$ is vacuously an annihilator of $s$.  For $1\leq i\leq n$, we write $s^{(i)}$ for $(s_1,\ldots,s_i)$.  If $n\geq 2$, then $\Ann(s)\subseteq \Ann(s^{(n-1)})$. 
If $d\leq n-1$ and the leading term of $f$ is a unit, we can make $f$ monic and generate the last $n-d$ terms of $s$ recursively from the first $d$ terms.

The following definition is a functional version of \cite[Definition 2.10]{N95b}.
\begin{definition} [Discrepancy Function] We define $\Delta: D[x]^\times\times D^n\ra D$ by 
$$\Delta(f,s)=(f\cdot\ul{s})_{\deg(f)-n}.$$
\end{definition}
Thus $\Delta(f,s)=\sum_{k=0}^{d}f_k \ s_{j-d+k}$ where $d=\deg(f)$.  Clearly for  $n\geq 2$, $f\in\Ann(s)^\times$  if and only if  $f\in\Ann(s^{(n-1)})^\times$ and $\Delta(s,f)=0$.

For any $s_1\in D^\times$ and  constant polynomial $f$, $\Delta(f,(s_1))=s_1$. If $s$ has exactly $n-1\geq 1$ leading zeroes, $s_n\neq 0$ and $f=1$, then $f\in\Ann(s^{(n-1)})$, but $\Delta(f,s)=s_n\neq 0$. Let $s$ be such that $s^{(n-1)}$ is geometric  with common ratio $r\in D^\times$, but $s$ is not geometric. In this case, we have $x-r\in \Ann(s^{(n-1)})$ but $\Delta(x-r,s)\neq 0$.  

If $s\in D^n$ is understood, we write $\Delta_n(f)$ for  $\Delta(f,s)$;  if $f$ is also understood, we simply write $\Delta_n$.
 It is elementary that  if $1\leq i\leq n-1$, then
$(f\cdot\ul{s^{(i)}})_{\deg(f)-i}=(f\cdot\ul{s})_{\deg(f)-i}$.

\section{Minimal Polynomials}\label{mptheory}
A notion of a 'minimal polynomial' of a finite sequence over a field  seems to have first appeared in \cite[Equation (3.16)]{R86}, where {\em the} minimal polynomial of a finite sequence was defined in terms of the output of the BM algorithm of \cite{Ma69}. We were unaware of \cite{R86} and adopted a more basic and more general approach which is independent of the BM algorithm. In particular, the approach introduced in \cite{N95b} is independent of  linear feedback shift registers and connection polynomials. For us,  a sequence may have more than one minimal polynomial.

\begin{definition}[Minimal Polynomial] \label{mindefn} (\cite[Definition 3.1]{N95b}) We say that $f\in \Ann(s)$ is a minimal polynomial of $s\in D^n$ if  $$\deg(f)=\min\{\deg(g):\ g\in \Ann(s)^\times\}$$ and let $\Min(s)$ denote the set of minimal polynomials of $s$.
\end{definition}
As any $f\in D[x]$ of degree at least $n$ annihilates $s\in D^n$, $\Min(s)\neq \emptyset$. 
We do not require minimal polynomials to be monic.
For any $d\in D^\times$, $d\in \Min(0,\ldots,0)$; if $s_1\neq 0$ and $\deg(f)=1$  then $f\in \Min((s_1))$ since $D$ has no zero divisors.  

The {\bf linear complexity function} $\LC:D^n\ra\{0\}\cup \N$ is $$\LC(s)=\deg(f)\mbox{ where }f\in\Min(s).$$
We will also write $\LC_n$ for $\LC(s)$ when $s$ is understood and similarly  $\LC_j=\LC(s^{(j)})$ for $1\leq j\leq n$.  For fixed $s$, $\LC$ is clearly a non-decreasing function of $i$. 

 It is trivial that if $s$  is infinite and satisfies a linear recurrence relation, then $$\Ann(s)\subseteq\bigcap_{n\geq 1} \Ann(s^{(n)}).$$ 
When $D$ is a field, a minimal polynomial of a linear recurring sequence $t$ is usually defined as a generator of the ideal $\Ann(t)$; see \cite[Chapter 8]{LN83}.

\begin{proposition} (Cf. \cite{S05}) Let $n\geq 1$, $s\in D^n$ and $f\in\Min(s)$  be monic. Define $t\in D^\N$ to be the extension of $s$ by $f$. If $\Ann(t)$ is principal then $\Ann(t)=(f)$.
\end{proposition}
\begin{proof} 
Let  $ \Ann(t)=(g)$ say. As $f\in\Ann(t)^\times$, $\Ann(t)\neq (0)$. If $g\neq 0$ generates $\Ann(t)$ then $g|f$ and $\deg(f)\geq \deg(g)$. Since
$g\in \Ann(s)$, we cannot have $\deg(g)<\deg(f)$, for  then $f\not\in\Min(s)$. So $\deg(g)=\deg(f)=d$ say. Equating leading coefficients shows that $g_d$ is a unit of $D$ and so we can also assume that $g$ is monic. We conclude that $f=g$ and that $\Ann(t)=(f)$.
\end{proof}

 It will follow from Proposition \ref{duality} below that  the (unique) minimal polynomial  of \cite{R86}  obtained from the output of the BM algorithm is an example of a minimal polynomial as per Definition \ref{mindefn}.  

 \subsection{Exponents}
The following definition will play a key role in defining our recursive minimal polynomial function.
The reason for choosing the term 'exponent' will become clear below.
\begin{definition}[Exponent Function] For $n\geq 1$, let the $n^{th}$ exponent function $\ee_n:D[x]^\times\ra \Z$ be given by 
$$\ee_n(f)=n+1-2\deg(f).$$
\end{definition}

 The following lemma is the annihilator analogue of \cite[Lemma 1]{Ma69} and will be used for proving minimality. We include a short proof  to keep the presentation self-contained. Commutativity and the absence of zero-divisors are essential here.
\begin{lemma}   \label{elegantproof}(\cite[Lemma 5.2]{N95b})  Let $n\geq 2$, $f\in\Ann(s^{(n-1)})^\times$ and $\Delta_n(f)\neq 0$. 

(i)  For any $g\in\Ann(s)^\times$, $\deg(g)\geq n-\deg(f)= \ee_{n-1}(f)+\deg(f)$. 

(ii) If $h\in\Min(s)$ then
$\deg(g)\geq\max\{\ee_{n-1}(h),0\}+\deg(h)$.
\end{lemma} 
\begin{proof} Put $\Delta=\Delta_n(f)$. We can write $f \cdot \ul{s}=N+\Delta \cdot x^{d-n}+P$ where  $d=\deg(f)$, $N_i=0$ for $d-n\leq i\leq -1$ and $P\in D[x]$. 
Likewise,  write $g \cdot \ul{s}=M+Q$ and  $e=\deg(g)$, with $M_i=0$ for $e-n\leq i\leq -1$ and $Q\in D[x]$. Let $h\in D[x]$ be $h=f  \cdot Q-g \cdot  P=g \cdot N-f \cdot  M+g \cdot \Delta \cdot   x^{d-n}$.  By construction $(g \cdot N-f \cdot  M)_{d+e-n}=0$, so $h_{d+e-n}=g_e \cdot \Delta \neq 0$ and $d+e-n\geq 0$.
The last sentence is immediate since $\Ann(s)\subseteq \Ann(s^{(n-1)})$ and  $\max\{\ee(f),0\}+\deg(f)=\max\{n-\deg(f),\deg(f)\}$.
\end{proof} 

If $s$ has exactly $n-1\geq 1$ leading zeroes and $s_n\neq 0$, then $1\in\Min(s^{(n-1)})$ and so $\LC(s^{(n-1)})=\LC(s_1)=1$.   Lemma \ref{elegantproof} implies that $\LC(s)\geq n$ and since any polynomial of degree $n$ is an annihilator, $\LC(s)=n$. 
For a geometric sequence $s^{(n-1)}$ over $D$ with common ratio $r\in D^\times$ such that $s$ is not geometric, we have $x-r\in \Min(s^{(n-1)})$ and $\Delta(x-r,s)\neq 0$.  By Lemma \ref{elegantproof}, we have $\LC(s)\geq n-1$. We will see that $\LC(s)= n-1$.

If $f^{(j)}\in\Min(s^{(j)})$ for $1\leq j\leq n-1$ and $\ee_{n-1}=e_{n-1}(f^{(n-1)})>0$ then $\LC_n\geq \LC_{n-1}+\ee_{n-1}$ by Lemma \ref{elegantproof}, and inductively,
$$ \LC_n\geq \LC_{1}+\sum_{\ee_{j-1}(f^{(j)})>0}\ee_{j-1}(f^{(j)}).$$
Theorem \ref{bit} will imply that this is actually an equality.
\section{A Recursive Minimal Polynomial Function}
 We will define a recursive minimal polynomial function $\mu:D^n \ra D[x]$.  But first we need the following function (which assumes that $\mu:D^{n-1} \ra D[x]$ has been defined).
We also set $\Delta_0=1$. 
\subsection{The Index Function}
\begin{definition}[Index Function]\label{indices}
Let $n\geq 1$ and $s\in D^n$. We set $\mu^{(0)}=1$ (so that $\Delta_{1}=\Delta_1(\mu^{(0)})=s_1$) and  $\ee_0=1$. 
Suppose that for $1\leq j\leq n-1$, $\mu^{(j)}\in \Min(s^{(j)})$ has discrepancy $\Delta_{j+1}$ and exponent $\ee_j$. We define  the index function
$$':\{0,\ldots,n\}\rightarrow \{-1,n-1\}$$ 
by $0'=-1$ and for $1\leq j\leq n-1$
$$j'=\left \{\begin{array}{ll}
 (j-1)' &\mbox{ if }\Delta_{j}=0\mbox{ or }(\Delta_{j}\neq 0 \mbox{ and }\ee_{j-1}\leq 0)\\
j-1 & \mbox{ if } \Delta_{j}\neq 0\mbox{ and }\ee_{j-1}>0.
\end{array}
\right.
$$
\end{definition}
Thus for example, $1'=-1$ if $s_1=0$ and $1'=0$ when $s_1\neq 0$ (since $\ee_0>0$). More generally, if $s$ has $n-1\geq 0$ leading zeroes, then $(n-1)'=\cdots=0'=-1$ and $n'=n-1$. 
\begin{example} In Table 1, $2'=1'=0$, $4'=3'=2$ and $5'=4$ and in Table 2, $1'=0$ and $4'=3'=2'=1$.  
\end{example}
It is trivial that $j'\leq j-1$ for $0\leq j\leq n$.
We will see that the $j$ for which $\Delta_j\neq 0$ and
$\ee_{j-1}>0$ are precisely those $j$ for which $\LC_j=\LC_{j-1}+\ee_{j-1}$; the linear complexity  has increased by  $\ee_{j-1}$.

The next result is essential.
\begin{proposition} \label{Delta} For $0\leq j\leq n$, $\Delta_{j'+1}\neq 0$.
\end{proposition}
\begin{proof}
We have $\Delta_0=1$. Inductively, assume that $\Delta_{k'+1}\neq 0$ for all $k$, $0\leq k\leq j-1$. If $\Delta_{j}=0$, then $\Delta_{j'+1}=\Delta_{(j-1)'+1}\neq 0$ by the inductive hypothesis. But if $\Delta_{j}\neq 0$ and $\ee_{j-1}\leq 0$, then 
$\Delta_{j'+1}=\Delta_{(j-1)'+1}\neq 0$ by the inductive hypothesis. Otherwise $\Delta_{j'+1}
=\Delta_j$ since $j'=j-1$ and we are done. 
\end{proof}
 
The definition of $j'$  as a maximum  $a_j$ in \cite{N95b}, \cite{N99b} required $j\geq 3$ and $\LC_{j-1}>\LC_1$. This in turn necessitated (i) defining $a_j$ separately when $n=1$ or
($n\geq 2$ and $\LC_{j-1}=\LC_1$ for $1\leq j-1\leq n-1$) and (ii) merging the separate constructions of minimal polynomials into a single construction.  Further, 
\cite[Proposition 4.1]{N95b} showed that the two notions  coincide, and required that $\LC_{-1}=\LC_{0}=0$.
\subsection{The Recursive Theorem}
Our goal in this subsection is to define a recursive function $$\mu:D^n\ra  D[x]$$
such that for all $s\in D^n$, $\mu(s)\in\Min(s)$.
When $s$ is understood, we will write $\mu^{(j)}$ for $\mu(s^{(j)})$.
A minimal polynomial of $s^{(1)}$ is clear by inspection, so we could use $n=1$ as the basis of the recursion, but with slightly more work, we will see that we can use $n=0$ as the basis.
\begin{definition}[Basis of the Recursion]\label{initialvalues} Recall that $0'=-1$ and $\Delta_0=1$. Let $\varepsilon\in D$ be arbitrary but fixed and $s\in D^n$. We put $\mu^{(-1)}=\mu(s,-1)=\varepsilon$  and $\mu^{(0)}=\mu(s,0)=1$. 
\end{definition}
 Thus the exponent of $\mu^{(0)}$ is $\ee_0=1$ and $\Delta_{1}=s_1$.
 It follows from Proposition \ref{Delta} and Lemma \ref{elegantproof} that
 $$\LC_j\geq \LC_{j'+1}\geq \max\{j'+1-\LC_{j'},\LC_{j'}\}\geq j'+1-\LC_{j'}.$$
 A key step in the proof of Theorem \ref{bit} is that the first and last inequalities are actually equalities.\\
 
\noindent {\bf Notation} To simplify Theorem \ref{bit}, we will use the following notation:

(i) $\mu'=\mu\circ\ '$ (where $\circ$ denotes composition)

(ii) $\LC'=\deg\circ \mu'$

(iii) $\Delta'=\Delta\circ (+1)\circ\ '\circ (-1)$, where $\pm 1$ have the obvious meanings. 

\noindent Thus $$\mu'^{(j)}=\mu^{(j')},\ \LC'_j=\LC_{j'}\mbox{ and }\Delta'_{j}=\Delta(\mu^{(k)},s^{(k+1)})=(\mu^{(k)}\cdot \ul{s})_{\LC_k-k-1}$$
 where $k=(j-1)'$. 

The definition of $\mu:D^n\ra D[x]$ in the following theorem was motivated in \cite{N99b}: given a minimal polynomial function $\mu:D^{n-1}\ra D[x]$, the theorem constructs $\mu:D^n\ra D[x]$ such that for all $s\in D^n$,  $\mu(s)\in\Min(s)$. We note that to verify that $\mu(s)\in\Ann(s)$, we first need $\deg(\mu(s))$.

\begin{theorem} (Cf. \cite{Ma69}, \cite[Section 9.6]{PW}) 
\label{bit} Let $n\geq 1$ and $s\in D^n$ and assume the initial values of Definition \ref{initialvalues}.  Define  $\mu^{(n)}$ recursively by$$\mu^{(n)}=\left\{\begin{array}{ll}
\mu^{(n-1)}& \mbox{ if } \Delta_n=0\\\\
\Delta'_{n}\cdot x^{\max\{\ee_{n-1},0\}}\ \mu^{(n-1)}-\Delta_{n}\cdot x^{\max\{-\ee_{n-1},0\}}\ \mu'^{(n-1)}
&\mbox{ otherwise.}\end{array}\right.
$$
If $\Delta_n=0$, clearly $\mu^{(n)}\in \Min(s)$, $\LC_n=\LC_{n-1}$ and $\ee_n=\ee_{n-1}+1$. If $\Delta_{n}\neq 0$ then
\begin{tabbing}
\hspace{1cm}\=(i) 
 $\deg(\mu^{(n)})=\max\{\ee_{n-1},0\}+\LC_{n-1}=n'+1-\LC'_{n}$\\ 
\>(ii) $\mu^{(n)}\in \Min(s)$\\
\>(iii) $\ee_n=-|\ee_{n-1}|+1$.
\end{tabbing}
\end{theorem}

\begin{proof}  We prove (i)   by induction on $n$. For $n=1$,  $\mu^{(1)}=x-\Delta_1\cdot\varepsilon$ and $\max\{\ee_0,0\}+\LC_0=1=\deg(\mu^{(1)})$.  As for the second equality,  $1'=0$ since $\ee_0=1>0$ and $1'+1-\LC'_{1}=1=\deg(\mu^{(1)})$. 
Suppose inductively that $n\geq 2$ and that (i) is true for $1\leq j\leq n-1$.  

  If $\ee_{n-1}\leq 0$ then $\mu^{(n)}=
\Delta'_{n}\cdot \ \mu^{(n-1)}-\Delta_{n}\cdot x^{-\ee_{n-1}}\ \mu'^{(n-1)}$. We have to show that
$-\ee_{n-1}+\LC'_{n-1}<\LC_{n-1}$. But $-\ee_{n-1}+\LC'_{n-1}$ is 
$$-(n-2\LC_{n-1})+\LC'_{n-1}=-n+2\LC_{n-1}+(n-1)'+1-\LC_{n-1}=\LC_{n-1}+(n-1)'+1-n$$
 by the inductive hypothesis and we know that $(n-1)'\leq n-2$ for all $n\geq 1$. Thus $-\ee_{n-1}+\LC'_{n-1}<\LC_{n-1}$ and $\ee_{n-1}\leq 0$ implies that $\deg(\mu^{(n)})=\LC_{n-1}$.
 
Suppose now that $\ee_{n-1}>0$. We have to show that $\deg(\mu^{(n)})=\ee_{n-1}+\LC_{n-1}$ i.e. that $\ee_{n-1}+\LC_{n-1}>\LC'_{n-1}$. But
$\ee_{n-1}+\LC_{n-1}=n-\LC_{n-1}>\LC_{n-1}$ since $\ee_{n-1}>0$ and $\LC_{n-1}\geq \LC'_{n-1}$ as $\LC$ is non-decreasing. Hence $\ee_{n-1}+\LC_{n-1}>\LC'_{n-1}$ and $\deg(\mu^{(n)})=\max\{\ee_{n-1},0\}+\LC_{n-1}$.

To complete (i), we have to show that $\deg(\mu^{(n)})=n'+1-\LC'_{n}$ if $\Delta_n\neq 0$. But if $\ee_{n-1}\geq 0$, 
then $n'=(n-1)'$ by definition and we have seen that 
$\deg(\mu^{(n)})=\LC_{n-1}$, so the result is trivially true in this case. If $\ee_{n-1}>0$, then $\deg(\mu^{(n)})=n-\LC_{n-1}$ and $n'=n-1$ by definition.
Hence $\deg(\mu^{(n)})=n'+1-\LC'_{n}$ and the induction is complete.\\

(ii) We first show inductively that $\mu^{(n)}\in\Ann(s)$.  If $n=1$ and $\Delta_1\neq 0$, then $\mu^{(1)}=x-\Delta_1\cdot\varepsilon\in\Min(s^{(1)})$.
Suppose inductively that $n\geq 2$, (ii) is true for $1\leq j\leq n-1$  and $\Delta_{n}\neq 0$.  From Part (i),
$d=\deg(\mu^{(n)})=\max\{\ee_{n-1},0\}+\LC_{n-1}\geq 0$. In particular, $\mu^{(n)}\neq 0$. We omit the proof that
$\mu^{(n)}\in \Ann(s^{(n-1)})$, showing only that $(\mu^{(n)}\cdot \ul{s})_{d-n} =0$.

Put $\mu=\mu^{(n-1)}$,  $\mu'=\mu'^{(n-1)}$, $e=\ee_{n-1}$,  $\LC=\LC_{n-1}$ and $\LC'=\LC'_{n-1}$. If $e\leq 0$, then $d=\LC$ and
$$(\mu^{(n)}\cdot \ul{s})_{\LC-n}
=\Delta'_n\cdot (\mu\cdot \ul{s})_{\LC-n}- \Delta_{n}\cdot (x^{-e}  \mu'\cdot \ul{s})_{\LC-n}=\Delta'_{n} \cdot \Delta_{n}-\Delta_{n}\cdot \Delta'_{n}=0$$
since $\LC-n+e=-\LC=\LC'-(n-1)'-1$ and  so $(\mu'\cdot \ul{s})_{\LC-n+e}=\Delta'_{n}$.  If  $e>0$, $d=n-\LC$ and 
$$(\mu^{(n)}\cdot \ul{s})_{d-n}  =(\mu^{(n)}\cdot \ul{s})_{-\LC} =\Delta'_{n} \cdot   (x^{e}\mu\cdot \ul{s})_{-\LC}- \Delta_{n}  
\cdot  ( \mu'\cdot \ul{s})_{-\LC}
  =\Delta'_n\cdot    \Delta_{n} - \Delta_{n} \cdot  \Delta'_n=0
  $$
since $-\LC-e=\LC-n$ and $-\LC=\LC'-(n-1)'-1$. Thus $\mu^{(n)}\in \Ann(s)$. \\

We complete the proof of (ii) by showing that $\mu^{(n)}\in \Min(s)$. We know that $\mu^{(1)}\in \Min(s^{(1)})$. For $n\geq 2$, we know from (i) that 
$\deg(\mu^{(n)})=\max\{\ee_{n-1},0\}+\LC_{n-1}$ which is $\max\{\ee_{n-1}(\mu^{(n-1)}),0\}+\deg(\mu^{(n-1)})$
 and therefore $\mu^{(n)}\in \Min(s)$ by Lemma \ref{elegantproof}.\\
 
(iii) We also prove this inductively. Suppose first that $n=1$ and 
$\Delta_1\neq 0$ . Then $\ee_1(\mu^{(1)})=2-2\cdot 1=0$ and since $\ee_0>0$, $\ee_1=-\ee_0+1=0$. Let $n\geq 2$ and $\Delta_n\neq 0$.
If $\ee_{n-1}\leq 0$, then $\ee_n(\mu^{(n)})=n+1-2\LC_n=
n+1-2\LC_{n-1}=\ee_{n-1}+1=\ee_n$, and if $\ee_{n-1}>0$ then $\ee(\mu^{(n)})=n+1-2(n-\LC_{n-1})=1-n+2\LC_{n-1}=1-\ee_{n-1}=-|\ee_{n-1}|+1=\ee_n$.
\end{proof}
\begin{remarks}
\begin{enumerate}
\item For  $\Delta_{n}\neq 0$ and $e=\ee_{n-1}$
$$\mu^{(n)}= \left\{\begin{array}{lll}

\Delta'_n \cdot  \mu^{(n-1)}- \Delta_{n} \cdot    x^{-e}  \mu'^{(n-1)}	& \mbox{if } e\leq 0\\\\
\Delta'_n \cdot  x^{+e}\mu^{(n-1)}- \Delta_{n}  \cdot  \mu'^{(n-1)} &\mbox{ if }e\geq 0.
\end{array}
\right.
$$\item If $s$ has precisely $n-1\geq 0$ leading zeroes, Theorem \ref{bit} yields $\mu^{(n)}=x^n-\varepsilon$.

\item We note that $\deg(\mu^{(n)})=n'+1-\LC'_n$ is trivially true if $\Delta_n=0$ or if $n=0$ (if we set $\LC'=0$).  We can also prove that  $\deg(\mu^{(n)})=n'+1-\LC'_n$ using Lemma \ref{elegantproof} (as in \cite{BL83}, \cite{Ma69} and \cite{N95b}) but prefer the simpler, direct argument used in Theorem \ref{bit}.

\item As noted in \cite{N95b}, we can use any $\mu^{(k)}$ instead of  $\mu^{(i)}$ (with appropriate powers of $x$) as long as
$\Delta_{k+1}\neq 0$ and $k<n-1$), but minimality is not guaranteed.
\end{enumerate}
\end{remarks}
\subsection{Some Corollaries}
 Let $n\geq 2$, $s\in D^n$ and $2\leq j\leq n$. Then $j$ is a {\bf jump point of $s$} if $\LC_j>\LC_{j-1}$. We write $\J(s)$ for the set of jump points of $s$. 
We do not assume that $\J(s)\neq \emptyset$. 
Evidently, the following are equivalent: (i) $j\in\J(s)$  (ii) $e_{j-1}>0$ (iii)  $j'=j-1$ (iv) $\LC_j=j-\LC_{j-1}>\LC_j$. The following is clear.
\begin{proposition}  For all $s\in D^n$, $\LC_n=\LC_1+\sum_{j\in\J(s)} \ee_{j-1}$. 
\end{proposition}
\begin{proof} Simple inductive consequence of Theorem \ref{bit}(i).
 \end{proof}
An important consequence of Theorem \ref{bit}(i) is the following well-known result.
\begin{corollary}\label{LL'} For any $s\in D^n$, $\LC_n=n'+1-\LC'_{n}$.
\end{corollary}
Next we use the index function to simplify the proof of \cite[Proposition 4.13]{N95b}.
\begin{proposition}\label{myprop}(Cf. \cite{Ma69}) Let $s\in D^n$. If  $f'\in D[x]$ and $\deg(f')\leq -\ee_n$, then $\mu^{(n)}+f' \mu'^{(n)}\in\Min(s)$. In particular, if $\ee_n\leq 0$ then $|\Min(s)|>1$.
\end{proposition} 
\begin{proof}  We will omit scripts.  By Corollary \ref{LL'}, $\LC=n'+1-\LC'$, so that
\begin{eqnarray}\label{enfin}
\deg(f'\mu')\leq -\ee+\LC'=2\LC-n-1+(n'+1-\LC)=\LC+n'-n\leq \LC-1
\end{eqnarray}
since $n'\leq n-1$, so  $\deg(\mu+f' \mu')=\LC$. 
Let $\LC-n\leq j\leq -1$. Now $$((\mu+f' \mu')\cdot\ul{s})_j=(\mu\cdot\ul{s})_j+(f' \mu'\cdot\ul{s})_j=(f' \mu'\cdot\ul{s})_j.$$
Inequality (\ref{enfin}) gives $\deg(f'\mu')-n'\leq \LC-n$, so we are done.
\end{proof}

It is convenient to introduce $\pp:\{0,\ldots,n-1\}\ra\N$
given by $$\pp(j)=\left\{\begin{array}{ll}
1 &\mbox{ if } j=0\\
j-j' & \mbox{ otherwise.}\end{array}\right.
$$
 It is clear that if $\Delta_n=0$, then $\pp(n)=\pp(n-1)+1$. We set $$\mu^{(n)\ \ast}=\ast\circ \mu(s,n)$$
 where $\ast$ denotes the reciprocal function and similarly
for fixed $s$, $\mu'^{(n)\ \ast}=\ast\circ \mu(s,n')$ .

\begin{corollary} \label{recip}
If $\Delta_{n}\neq 0$ then
\begin{tabbing}
\hspace{1.5cm}\= (i) $\mu^{(n)\ \ast}=\Delta'_{n}\cdot \mu^{(n-1)\ \ast}- \Delta_{n} \cdot
x^{\pp(n-1)}\ \mu'^{(n-1)\ \ast} $\\\\ 
  
\> (ii) $\pp(n)=\pp(n-1)+1$ if $\ee_{n-1}\leq 0$ and $\pp(n)=1$  otherwise.
\end{tabbing}
\end{corollary}
\begin{proof}  Put $\mu=\mu^{(n-1)}$, 
$e=\ee_{n-1}$, $\LC=\LC_{n-1}$, $\pp=\pp(n-1)$, $\mu'=\mu'^{(n-1)}$ and $\LC'=\LC'_{n-1}$.  If $e\leq 0$, then
$\mu^{(n)}=\Delta'_{n}\cdot \mu-\Delta_{n} \cdot x^{-e}\ \mu'$ and $\LC_n=\LC$. Then
\begin{eqnarray*} 
\mu^{(n)\ast}&=&x^{\LC_n}\mu^{(n)}(x^{-1})=\Delta'_n \cdot x^\LC\  \mu(x^{-1})-\Delta_{n}\cdot  x^{\LC+e}\ \mu'(x^{-1})\\
&=& \Delta'_{n}\cdot \mu^\ast-\Delta_{n}\cdot  x^{\LC+e-\LC'}\  \mu'^\ast=\Delta'_n \cdot \mu^\ast-\Delta_{n}\cdot  x^\pp\  \mu'^\ast
\end{eqnarray*}
since $\LC+e-\LC'=n-\LC-\LC'=n-1-(n-1)'=\pp$ by Corollary \ref{LL'}. If $e>0$,
$\mu^{(n)}=\Delta'_n\cdot  x^{e}\ \mu-\Delta_{n}\cdot  \mu'$ and $\LC_n=n-\LC$, so
\begin{eqnarray*}
\mu^{(n)\ast}&=&x^{\LC_n}\ \mu^{(n)}(x^{-1})=\Delta'_{n}\cdot  x^{n-\LC-e}\  \mu(x^{-1})-\Delta_{n}\cdot  x^{n-\LC}\ \mu'(x^{-1})\\
&=&\Delta'_{n}\cdot  \mu^\ast-\Delta_{n}\cdot  x^{n-\LC-\LC'} \ \mu'^\ast= \Delta'_{n}\cdot \mu^\ast-\Delta_{n}\cdot  x^\pp\  \mu'^\ast
\end{eqnarray*}
since $n-\LC-e=\LC$ and $n-\LC-\LC'=n-(n-1)'-1=\pp$.   
The value of $\pp(n)$ is immediate from
Corollary \ref{LL'}.
\end{proof}
 \subsection{The Iterative Version}
 We could obtain $\mu^{(n)}$  recursively using Theorem \ref{bit}, but it is more efficient to obtain it iteratively. 
\begin{corollary}[Iterative Form of $\mu$]\label{noindices}
Let $n\geq 1$, $s\in D^n$ and $\varepsilon\in D$.  
Assume the initial values of Definition \ref{initialvalues}. For $1\leq j\leq n$, let
$$\mu^{(j)}=\left\{\begin{array}{ll}
\mu^{(j-1)} &\mbox{ if } \Delta_{j}=0\\\\
\Delta'_{j} \cdot  x^{\max\{\ee_{j-1},0\}}  \mu^{(j-1)}- \Delta_{j} \cdot    x^{\max\{-\ee_{j-1},0\}}  \mu^{'({j-1})}&\mbox{ otherwise.}
\end{array}\right.
$$
Then $\mu^{(j)}\in\Min(s^{(j)})$. Further, if $\Delta_j=0$, then $\mu^{'(j)}=\mu^{'(j-1)}$,  $\Delta'_{j+1}=\Delta'_{j}$ and $\ee_j=\ee_{j-1}+1$.
If $\Delta_j\neq 0$ then 
\begin{tabbing}\hspace{1cm}\= 
(a) if $\ee_{j-1}\leq 0$ then $\mu^{'(j)}=\mu^{'(j-1)}$ and $\Delta'_{j+1} =\Delta'_{j}$ \\
\>(b) but if $\ee_{j-1}>0$ then $\mu^{'(j)}=\mu^{(j-1)}$ and $\Delta'_{j+1} =\Delta_{j}$ \\
\>(c) $\ee_j=-|\ee_{j-1}|+1$.
\end{tabbing}
\end{corollary}
In other words, when  $\Delta_{j}\neq 0$,
$$\mu^{(j)}= \left\{\begin{array}{lll}

\Delta'_j \cdot  \mu^{(j-1)}- \Delta_{j} \cdot    x^{-e}  \mu'^{(j-1)}	& \mbox{if } e=\ee_{j-1}\leq 0\\\\
\Delta'_j \cdot  x^{+e}\mu^{(j-1)}- \Delta_{j}  \cdot  \mu'^{(j-1)} &\mbox{otherwise.}
\end{array}
\right.
$$

 We are now ready to derive an algorithm to compute a minimal polynomial for $s\in D^n$ from Corollary \ref{noindices}. The initialisation is clear. Let $1\leq j\leq n-1$. From the definition of $\ee_{j-1}=\ee_{j-1}(\mu^{(j-1)})$, we have $\LC_{j-1}=\frac{j-\ee_{j-1}}{2}$  and so $$\Delta_j = (\mu^{(j-1)}\cdot \ul{s^{(j)}})_{\LC_{j-1}-j}=\sum_{k=0}^{\frac{j-\ee_{j-1}}{2}} \mu^{(j-1)}_k   s_{k+\frac{j+\ee_{j-1}}{2}}.$$ 
For the body of the loop, we next show how to suppress $j-1$ and $j$.  When $\Delta_j=0$, we ignore the updating of $\mu^{(j-1)}$, $\mu^{'(j-1)}$ and $\Delta'_{j}$, but $\ee_j=\ee_{j-1}+1$. But when  $\Delta_{j}\neq 0$, (a)
$$\mu^{(j)}= \left\{\begin{array}{lll}

\Delta' _{j}\cdot  \mu^{(j-1)}- \Delta_{j} \cdot    x^{-e}  \mu'^{(j-1)}	& \mbox{if } e\leq 0\\\\
\Delta'_{j} \cdot  x^{+e}\mu^{(j-1)}- \Delta_{j}  \cdot  \mu'^{(j-1)}&\mbox{otherwise.}
\end{array}
\right.
$$
and (b) we need to update  $\mu^{'(j-1)}$ and $\Delta'_{j}$ when  $\ee_{j-1}>0$; since (a) will overwrite $\mu^{(j-1)}$ once we have suppressed $j$, we  keep a copy $t$ of $\mu^{(j-1)}$ when  $\ee_{j-1}>0$, so that the updating is $\mu^{'(j)}=t$ and $\Delta'_{j+1}=\Delta_{j}$.
For (c), we have
$\ee_j=\ee_{j-1}+1$ if $\ee_{j-1}\leq 0$ and $\ee_j=-\ee_{j-1}+1$ otherwise. Now only the current values of the variables appear and so we can suppress scripts.
  The following algorithm (written in the style of \cite{Wirth}) is now immediate.
\begin{algorithm}[Iterative minimal polynomial]   \label{rewrite}\ 
\begin{tabbing}

\noindent Input: \ \ \=$n\geq 1$, $\varepsilon\in D$ and $s=(s_1,\ldots,s_{n})\in D^n$.\\

\noindent Output: \>$\mu\in\Min(s)$.\\\\

\{$e := 1$;\ $\mu':=\varepsilon$;\ $\Delta':=1$;
$\mu  :=  1;\ $\\
{\tt FOR} \= $j = 1$ {\tt TO }$n$\\
    \> \{$\Delta    :=  \sum_{k=0}^{\frac{j-e}{2}} \mu_k \  s_{k+\frac{j+e}{2}};$ \\
   \> {\tt IF } $\Delta  \neq  0$ {\tt THEN }\{{\tt IF } $e\leq 0$  \=   						{\tt THEN } $\mu   :=  \Delta'\cdot \mu-\Delta\cdot  x^{-e} \mu'$;\\\\
  \>                 \>    {\tt ELSE} \{\=$t :=  \mu$;\\
  \>                 \>           \> $\mu  :=   \Delta'\cdot x^e\mu-\Delta\cdot  \mu'$;\\
  \>                 \>            \> $\mu':= t$; \ $\Delta':= \Delta$;\\ 
    \>                 \>            \>$e := -e$\}\}\\
  \> $e  := e+1$\}\\
{\tt RETURN}$(\mu)$\}
\end{tabbing} 
\end{algorithm}
\begin{example} 
Tables 1 and 2 give the values of $e$ and $\Delta$, and outputs $\mu$, $\mu'$ for the binary sequence  (1,0,1,0,0) of \cite{Ma69}  and for the integer sequence  (0,1,1,2), with $\varepsilon=0$ in both cases. 
\end{example}
\begin{table}   \label{shortexMP}
\caption{Algorithm MP with $\varepsilon=0$, input $(1,0,1,0,0)\in \mathrm{GF}(2)^{5}$}
\begin{center}
\begin{tabular}{|c|c|c|l|l|l|l|}\hline
$j$ & $\ee_{j-1}$  & $\Delta_{j}$   			&$\mu^{(j)}$ & $\mu^{'(j)}$ \\\hline\hline
 $1$   &$1$ &$1$ & $x$ &$1$ \\\hline
$2$    & $0$  &$0$  & $x$ & $1$ \\\hline
$3$   & $1$   &$1$ &  $x^2+1$& $x$ \\\hline
$4$     &$0$  &$0$  & $x^2+1$& $x$\\\hline
$5$   & $1$ &$1$  & $x^3$&  $x^2+1$.\\\hline
\end{tabular}
\end{center}
\end{table}

\begin{table}   \label{fibMP}
\caption{Algorithm MP with $\varepsilon=0$, input $(0,1,1,2)\in \Z^4$}
\begin{center}
\begin{tabular}{|c|c|c|l|l|l|l|}\hline
$j$ & $\ee_{j-1}$  & $\Delta_{j}$  		&$\mu^{(j)}$ & $\mu^{'(j)}$ \\\hline\hline
 $1$   &$1$ &$0$  & $1$ &$0$ \\\hline
$2$    & $2$  &$1$  & $x^2$ & $1$ \\\hline
$3$   & $-1$   &$1$ &  $x^2-x$& $1$ \\\hline
$4$     &$0$  &$1$  & $x^2-x-1$& $x-1$.\\\hline
\end{tabular}
\end{center}
\end{table}

\section{A Recursive BM Theorem}
\subsection{Reciprocal Pairs}
\begin{definition}[Reciprocal Pair] Let $n\geq 1$ and $s\in D^n$. We say that $(g,\ell)\in D[x]\times [0,n]$ is a {reciprocal} {pair} for $s$, written $(g,\ell)\in \Rp(s)$, if $g_0\neq 0$, $d=\deg(g)\leq \ell$ and
 $\ell+1\leq j\leq n$ implies that 
\begin{equation}   \label{feedback}
(g \cdot\overline{s})_j=g_0s_ j+ g_1s_{ j-1}+\cdots +g_d s_{ j-d}=0.
\end{equation}
\end{definition}
For $n\geq 2$, the $n^{th}$ {\em discrepancy} of $(g,\ell)\in \Rp(s^{(n-1)})$ is
$\Delta_n(g,\ell)=(g\cdot\ol{s})_n=\sum_{k=0}^d g_k s_{n-k}$,
 and $(g,\ell)\in \Rp(s)$ if and only if $\Delta_n(g,\ell)=0$. Note that $\ell$ is often used instead of $\deg(g)$ in the sum of Equation (\ref{feedback}) and in  the discrepancy \cite{Ma69}; we prefer to use $\deg(g)$ since $g$ is then a genuine polynomial.

\begin{proposition} \label{duality} Let $s\in D^n$, $f\in D[x]$ and $d=\deg(f)\geq 0$.  Then for $d+1\leq j\leq n$, $(f\cdot \ul{s})_{d-j}=(f^\ast\ol{s})_j$. 
Thus if $f\in \Ann(s)^\times$, $(f^\ast,d)\in \Rp(s)$ and if $(g,\ell)\in \Rp(s)$, then 
$x^{\ell-\deg(g)}g^\ast\in \Ann(s)^\times$.
\end{proposition}
\begin{proof} We have $f^\ast(x)=x^df(x^{-1})$, so
$f(x^{-1})=x^{-d}f^\ast(x)$ and $f(x)=x^df^\ast(x^{-1})$. Hence 
$(f(x)\cdot \ul{s})_{d-j}=(f^\ast(x^{-1})\cdot \ul{s})_{-j}=(f^\ast\cdot\ol{s})_j$.
\end{proof}
In particular, $\Delta_n(\mu^{(n-1)})=\Delta_n(\mu^{(n-1)\ast},\deg(\mu^{(n-1)}))$ and using $\Delta_n$ in two ways causes no confusion.

\subsection{Shortest Reciprocal Pairs}
\begin{definition}[Shortest Reciprocal Pair] Let $s\in D^n$. We say that a reciprocal pair $(g,\ell)$ for $s$ is {shortest}, written $(g,\ell)\in \SRp$, if $x^{\ell-\deg(g)}g^\ast\in \Min(s)$.
\end{definition}
Note that when $x^{\ell-\deg(g)}g^\ast\in \Min(s)$, $\ell=\deg(x^{\ell-\deg(g)}g^\ast)=\LC(s)$ since $g_0\neq 0$.
We define the index  function exactly as in the minimal polynomial case and set $$\ee_n( \varrho^{(n)},\LC_n)=n+1-2\LC_n.$$

\begin{theorem}[Recursive BM] \label{one-step BM} Let $n\geq 1$ and $s\in D^n$. Put $\varrho^{(-1)}=\varepsilon$ and $\Delta_0=1$. Define $\varrho^{(n)}$ recursively by 

$$\varrho^{(n)}=\left\{\begin{array}{ll}
\varrho^{(n-1)}& \mbox{ if } \Delta_n=0\\\\
\Delta'_{n}\cdot x^{\max\{\ee_{n-1},0\}}\ \varrho^{(n-1)}-\Delta_{n}\cdot x^{\max\{-\ee_{n-1},0\}}\ \varrho'^{(n-1)}
&\mbox{ otherwise.}\end{array}\right.
$$
If $\Delta_n=0$, clearly $\varrho^{(n)}\in \SRp(s)$, $\LC_n=\LC_{n-1}$ and $\ee_n=\ee_{n-1}+1$. If $\Delta_{n}\neq 0$ then
\begin{tabbing}
\hspace{1cm}\=(i) 
 $\LC_n=\max\{\ee_{n-1},0\}+\LC_{n-1}=n'+1-\LC'_{n}$\\ 
\>(ii) $\varrho^{(n)}\in \SRp(s)$\\
\>(iii) $\ee_n=-|\ee_{n-1}|+1$.
\end{tabbing}
\end{theorem}
\begin{proof} We suppose that $\Delta_n\neq 0$. Let $\mu^{(n-1)}=x^{\LC_{n-1}-\deg(\varrho^{(n-1)})}\varrho^{(n-1)\ast}\in\Min(s^{(n-1)})$ by Proposition \ref{duality}. Now  let $\mu^{(n)}\in\Min(s)$ be as in Theorem \ref{bit}. Further,   $\LC_n=\max\{\ee_{n-1},0\}+\LC_{n-1}$ and $\ee_n=n+1-\LC_n=-|\ee_{n-1}|+1$. By Corollary \ref{recip}, $\mu^{(n)\ast}=\varrho^{(n)}$ and  $(\varrho^{(n)},\LC_n)
\in \SRp(s)$, which completes the proof.
\end{proof}
\subsection{Iterative BM}
As before, it is convenient to write $\varrho'^{(n)}=\varrho^{(n')}$.

\begin{corollary}[Iterative BM]\label{noindicesrecip}
Let $n\geq 1$, $s\in D^n$ and $\varepsilon\in D$. Put $\varrho^{(0)}=1$, $\ee_0=1$,  $\varrho^{'(0)}=\varepsilon$ and $\Delta'_{0}=1$.  For $1\leq j\leq n$, let
$$\varrho^{(j)}=\left\{\begin{array}{ll}
\varrho^{(j-1)} &\mbox{ if } \Delta_{j}=0\\\\
\Delta'_{j} \cdot  \varrho^{(j-1)}- \Delta_{j} \cdot    x^{p(j-1)}  \varrho^{'({j-1})}&\mbox{ otherwise.}
\end{array}\right.
$$
Then $\varrho^{(j)}\in\SRp(s^{(j)})$. Further, if $\Delta_j=0$ then $\varrho^{'(j)}=\varrho^{'(j-1)}$,  $\Delta'_{j+1}=\Delta'_{j}$, $\ee_j=\ee_{j-1}+1$ and $p(j)=p(j-1)+1$.
If $\Delta_j\neq 0$ then 
\begin{tabbing}\hspace{1cm}\= 
(a) if $\ee_{j-1}\leq 0$ then $\varrho^{'(j)}=\varrho^{'(j-1)}$,  $\Delta'_{j+1} =\Delta'_{j}$ and $p(j)=p(j-1)+1$\\
\>(b) but if $\ee_{j-1}>0$ then $\varrho^{'(j)}=\varrho^{(j-1)}$, $\Delta'_{j+1} =\Delta_{j}$ and $p(j)=1$ \\
\>(c) $\ee_j=-|\ee_{j-1}|+1$.
\end{tabbing}
\end{corollary}

As for the minimal polynomial case, Corollary \ref{noindicesrecip} immediately yields an algorithm.  The only difference is that we now have a single expression for
$\varrho^{(n)}$ (which we can factor out), we begin with $p=1$ and we set $p=0$ if ($\Delta_n\neq 0$ and $e>0$) --- so that we always increment $p$ by 1. 
\begin{algorithm} [Iterative BM]  \label{newBMa}\ 
(Cf.  \cite[Algorithm 2.2]{Ma69})
\begin{tabbing}
\noindent Input: \ \ \=$n\geq 1$, $\varepsilon\in D$,  and $s=(s_1,\ldots,s_{n})\in D^n$.\\

\noindent Output: \> $(\varrho,\LC)\in \SRp$ i.e. $\varrho_0\neq 0$ and $x^{\LC-\deg(\varrho)}\varrho^\ast\in \Min(s)$.\\\\
\{$e:=1$;\  $\varrho':=\varepsilon$:\ $\Delta ':=1$;\ $p:=1$;\  $\varrho :=  1$; \\

  {\tt FOR} \= $j = 1$ {\tt TO }$n$\\
    \> \{$\Delta     :=  \sum_{k=0}^{\deg(\varrho)}\varrho_k\  s_{j-k};$ \\
   \> {\tt IF } $\Delta   \neq  0$ {\tt THEN }\{\=$t:=\varrho$; $\varrho   :=  \Delta '\cdot \varrho-\Delta \cdot x^p\varrho'$;\\
  \> \>   {\tt IF }$e>0$ {\tt THEN }\= \{$\varrho' :=  t$; $\Delta':=\Delta$;    
  $p:=0$; $e:=-e$\}\}\\
  \> $p := p+1$; $e:=e+1$\}\\\\
 {\tt RETURN}$(\varrho,\frac{n+1-e}{2})$\}
 \end{tabbing}

\end{algorithm}

If $D$ is a field, $\varepsilon=1$ and we make each $\varrho$ monic then Algorithm \ref{newBMa} is equivalent to the LFSR synthesis algorithm of \cite[p. 124]{Ma69} (replace $e$ by $j-2\LC$, $\deg(\varrho)$ by $\LC$ and relabel the variables).  

\begin{example} 
Tables 3 gives the values of $e$, $\Delta$, $p$ and the outputs $\varrho$, $\varrho'$ for the binary sequence  (1,0,1,0,0) of \cite{Ma69}. Table 4 gives similar information for  the integer sequence  (0,1,1,2). In both cases $\varepsilon=0$. 
\end{example}
\begin{table}   \label{shortexM}
\caption{Algorithm BM with $\varepsilon=0$, input $(1,0,1,0,0)\in \mathrm{GF}(2)^{5}$}
\begin{center}
\begin{tabular}{|c|c|c|c|l|l|l|l|}\hline
$j$ & $\ee_{j-1}$  & $\Delta_{j-1}$  &$p_{j-1}$   			&$\varrho^{(j)}$ & $\varrho^{'(j)}$ \\\hline\hline
 $1$   &$1$ &$1$ &$1$ & $1$ &$1$ \\\hline
$2$    & $0$  &$0$  &$1$& $1$ & $1$ \\\hline
$3$   & $1$   &$1$&$2$ &  $1+x^2$& $1$ \\\hline
$4$     &$0$  &$0$  &$1$& $1+x^2$& $1$\\\hline
$5$   & $1$ &$1$  &$2$& $1$&  $1+x^2$.\\\hline
\end{tabular}
\end{center}
\end{table}

\begin{table}   \label{fib}
\caption{Algorithm BM with $\varepsilon=0$, input $(0,1,1,2)\in \Z^4$}
\begin{center}
\begin{tabular}{|c|c|c|c|l|l|l|l|}\hline
$j$ & $\ee_{j-1}$  & $\Delta_{j-1}$  &$p_{j-1}$   			&$\varrho^{(j)}$ & $\varrho^{'(j)}$ \\\hline\hline
 $1$   &$1$ &$0$ &$1$ & $1$ &$0$ \\\hline
$2$    & $2$  &$1$  &$2$& $1$ & $1$ \\\hline
$3$   & $-1$   &$1$&$1$ &  $1-x$& $1$ \\\hline
$4$     &$0$  &$1$  &$2$& $1-x-x^2$& $1-x$.\\\hline
\end{tabular}
\end{center}
\end{table}
\begin{remark} \label{IY} A  more complicated BM algorithm over a field (derived from properties of Hankel matrices) appears in \cite[p. 148]{IY}.  Indeed, the algorithm of \cite{IY}

(i) does not use the initial values of Corollary \ref{noindicesrecip}, but has several initialization steps 

(ii) uses a variable called $\Delta\LC$ which equals $\ee$, but  $\Delta\LC$ is not updated  incrementally 

(iii) uses a variable $k(j)$ defined as in \cite{Ma69} rather than using $j'$ and $\pp_{j}=j-j'$ 

(iv) does not maintain variables $\varrho'$ and $\Delta'$.
\end{remark} 

\section{Complexity of the Iterative Algorithms}
It is straightforward to show that at most $\frac{n(3n+1)}{2}$  multiplications in $D$ are required for Algorithm \ref{rewrite}, \cite[Proposition 3.23]{N95b}. In this section we show that this can be replaced by $3\lfloor \frac{n^2}{4}\rfloor$. 
\subsection{The Linear Complexity Sum}
We continue the previous notation: $\mu^{(0)}=1$, $\LC_{0}=0$, $\ee_0=1$ and for $0\leq j\leq n$, $\Delta_{j+1}=\Delta(\mu^{(j)})$ and $\ee_j=j+1-2\LC_j$.

The main result of this subsection uses the following lemma.\begin{lemma}\label{ltwo}
 For integers $u\geq 0$ and $t\geq 1$,
$\sum_{j=2u+1}^{2u+2t}\lfloor  \frac{j+1}{2}\rfloor=2tu+t^2$.
\end{lemma}
\begin{proof}
Put $w=2u+t+1$. The sum is
\begin{eqnarray*}\label{two}\sum_{k=0}^{t-1}\left( \lfloor \frac{w-k}{2}\rfloor + 
			\lfloor \frac{w+k+1}{2}\rfloor\right)=\sum_{k=0}^{t-1}\left( \frac{w-k}{2} + 
			\frac{w+k+1}{2}-\frac{1}{2}\right)=tw
\end{eqnarray*}
since $w-k$ and $w+k+1$ have opposite parity. \end{proof}

 \begin{lemma}\label{stable}
$\sum_{i=1}^{n}\LC_i \leq \sum_{i=1}^{n}\lfloor  \frac{i+1}{2}\rfloor$.
\end{lemma}
\begin{proof} Let us call $j\geq 0$ {\em stable}  if it is even, $\LC_j= \frac{j}{2}$ and $\sum_{i=1}^{j}\LC_i \leq \sum_{i=1}^{j}\lfloor   \frac{i+1}{2}\rfloor$.
Clearly $0$ is stable, so suppose inductively
that $2u\geq 0$ is stable. In particular, $\LC_{2u}=u$ and $\LC_{2u+1}=u$ independently of $\Delta_{2u+1}$. If $\Delta_{2u+2}\neq 0$ then $\LC_{2u+2}=u+1=
\lfloor  \frac{2u+3}{2}\rfloor$ and we can replace $u$ by $u+1$. Hence we can
assume that $\Delta_{2u+2}=0$, and that 
$\LC_{2u+1}=\cdots =\LC_{2u+t}=u$ for some maximal
$t$ such that $2u+2\leq 2u+t\leq n$. If $2u+t=n$, we are done since the result holds by the inductive hypothesis.

If $2u+t< n$, we show that there is a maximal stable $j_{\mathrm{M}}\leq n$. First we show that if $v=2u+2t\leq n$, then $v$ is stable.
  We have $\LC_{2u+t+1}\neq u$ and so $\Delta_{2u+t+1}\neq 0$  since $t$ is maximal. Hence $\LC_{2u+t+1}=u+t$. An easy induction shows that
$\LC_{2u+t+j}=\LC_{2u+t+j+1}$ for $1\leq j\leq t$ i.e. that $\LC_v=\LC_{2u+2t}=\LC_{2u+t+1}=u+t=\lfloor\frac{v+1}{2}\rfloor$.
Since $2u$ is stable, it is enough to show that
$\sum_{j=2u+1}^{v}\LC_j= \sum_{j=2u+1}^{v}\lfloor  \frac{j+1}{2}\rfloor$.
The left-hand-side is $tu+t(u+t)$ which equals the right-hand side by Lemma \ref{ltwo}(ii). So $v$ is stable.  
By induction there is a maximal stable $j_{\mathrm{M}}\leq n$. 

If $j_{\mathrm{M}}=n$, we are done. If $j_{\mathrm{M}}<n$, write $n=2u+t+1+m$ for $0\leq m< t-1$. It is enough to show that
$\sum_{i=2u+1}^{n}\LC_i
 \leq \sum_{i=2u+1}^{n}\lfloor  \frac{i+1}{2}\rfloor$ since $2u$ is stable.
Write the left-hand side as
$$\sum_{k=0}^{m} \LC_{2u+t-k}+\sum_{k=0}^{m} \LC_{2u+t+k+1}
+\sum_{i=2u+1}^{2u+t-m-1}\LC_i$$
The first summand is $(m+1)u$ and  the second
is $(m+1)(u+t)$. For $\sum_{i=2u+1}^{n}\lfloor  \frac{i+1}{2}\rfloor$, we proceed as in Lemma \ref{ltwo}(ii) using the pairs with indices $2u+t-k,2u+t+k+1$ for $k=0,\ldots,m$, while each of the terms in the third summand have $\LC_i=u$, which  is less or equal to
the corresponding $\lfloor\frac{i+1}{2}\rfloor$.
\end{proof} 

The following Corollary appeared in \cite{FJ98} for $n$ even.
\begin{corollary}    \label{sum.l} $\sum_{i=1}^{n}\LC_i \leq  \lfloor (n+1)^2/4\rfloor$.
\end{corollary}
\begin{proof} We have 
  $\sum_{j=1}^{n}\lfloor  \frac{j+1}{2}\rfloor= \lfloor (n+1)^2/4\rfloor$.
  \end{proof}

It turns out that sequences with a perfect linear complexity profile show that the upper bound of Corollary \ref{sum.l} is tight.   Recall that $s$ has {\em a perfect linear complexity profile (PLCP)} if   $\LC_j = \lfloor \frac{j+1}{2}\rfloor$ for $1\leq j\leq n$ \cite{R86}.
(This definition was initially given for binary sequences, but by Theorem \ref{bit}, it extends to sequences over $D$.)
\begin{proposition} \label{basictfae}The following are equivalent.

(i) $s$ has a PLCP

(ii) for $1\leq j\leq n$
$$\ee_j= \left\{\begin{array}{rl}

        1 & \mbox{if } j \mbox{ is even}\\
	0 &\mbox{otherwise}
 \end{array}
\right. $$

(iii) $\Delta_{j}\neq 0$ for all odd $j$, $1\leq j\leq n+1$.
\end{proposition}
\begin{proof}
(i) $\Leftrightarrow$ (ii): Easy consequence of the definitions.

(i) $\Rightarrow$ (iii): If $j\leq n+1$ is odd then $\Delta_j\neq 0$, for otherwise $\frac{j-1}{2}+1=\frac{j+1}{2}=\LC_{j}=\LC_{j-1}=\frac{j-1}{2}$. 

(iii) $\Rightarrow$ (i): Let  $\Delta_{j}\neq 0$ for all odd $j$, $1\leq j\leq n+1$. Then $s_1\neq 0$, $\LC_1=1$ and $\ee_1=0$. If $\Delta_2=0$, then $\LC_2=\LC_1=1$, otherwise $\LC_2=\max\{\ee_1,0\}+1=1$, so that $\LC_2$ is as required. Suppose that $j\leq n$ is odd and $\LC_{k}=\lfloor\frac{k+1}{2}\rfloor$ for all $k$, $1\leq k\leq j-1$. We have $\LC_j=j-\LC_{j-1}=j-\frac{j-1}{2}=\lfloor\frac{j+1}{2}\rfloor$. If $j=n+1$, we are done. Otherwise,  if $\Delta_{j+1}=0$, we have $\LC_{j+1}=\LC_j=\lfloor\frac{j+1}{2}\rfloor=\lfloor\frac{j+2}{2}\rfloor$, whereas if $\Delta_{j+1}\neq 0$,
$\LC_{j+1}=j+1-\LC_j=j+1-\lfloor\frac{j+1}{2}\rfloor=\lfloor\frac{j+2}{2}\rfloor$.
\end{proof}
It  follows that if $s$ has a PLCP, then $\sum_{j=1}^{n}\LC_j =  \lfloor (n+1)^2/4\rfloor$. In particular, this is true if $\Delta_j$ is always non-zero. Note that $\LC_j \leq \lfloor  \frac{j+1}{2}\rfloor$  does not hold in general: consider $(0,\ldots,0,1)\in D^n$ where $n\geq 2$ for example. 

We do not know if $\sum_{j=1}^{n}\LC_j =  \lfloor (n+1)^2/4\rfloor$ implies that $s$ has a PLCP.

\subsection{Worst-case Analysis}
It is now immediate that
\begin{theorem} \label{mult} For a sequence of $n$ terms from $D$, Algorithms   \ref{rewrite} and \ref{newBMa} require at most $3\lfloor \frac{n^2}{4}\rfloor$  multiplications in $D$.
\end{theorem}

As remarked above, if $D$ is a field then we can divide $\mu$ in Algorithm \ref{rewrite} and $\varrho$ in Algorithm \ref{newBMa} by $\Delta'$, making each polynomial monic. If we ignore the number of field divisions, this gives at most $2\lfloor \frac{n^2}{4}\rfloor$ multiplications. 
We note that  an upper bound of $\frac{n(n+1)}{2}$ for the maximum number of multiplications in the BM algorithm appeared in \cite[p. 209A]{Gus76}. 

\subsection{Average Analysis}
An average analysis of the BM algorithm appeared in \cite[Equation (15), p. 209]{Gus76} and was based on Proposition 1, {\em loc. cit.}, which was proved using the BM algorithm under the hypothesis that 'there is one formula for a sequence of length 
zero'. Another proof derived from the number of sequences with prescribed linear complexity and prescribed jump complexity appeared in \cite[Corollary 1]{Nied90}. 

We give a direct inductive proof of \cite[Proposition~1]{Gus76} which is independent of any particular algorithm. 
In particular, Theorem \ref{count} applies to Algorithm \ref{rewrite} and to Algorithm \ref{newBMa}.
 One could in principle set up and solve recurrence equations similiar to \cite[Equations (9), (10), (11)]{Gus76} to carry out an average analysis of Algorithms \ref{rewrite} and \ref{newBMa}, but we will do not do this here.
 
\begin{theorem} \label{count} Let $D=\F_q$. The number of sequences of length $n$ with linear complexity $\ell$ is 

$$\left\{   \begin{array}{ll}
           0                & \mbox{ if } \ell<0\\
           1          & \mbox{ if } \ell=0\\
           q^{2\ell-2n-1}(q-1) & \mbox{ if } 1 \leq \ell \le \lfloor n/2 \rfloor\\
           q^{2n-2\ell}(q-1)     & \mbox{ if } \lfloor n/2 \rfloor < \ell \le n\\
           0                & \mbox{ if } \ell>n.
                         \end{array}
          \right.$$
\end{theorem}
\begin{proof}  Put 
$N(n,\ell)=|\{s\in D^n \ : \LC_n=\ell\}|$. It is clear that  $N(n,\ell)$ is as stated for $\ell < 0$ or $\ell > n$.
We will show by induction on $n$ that $N(n,\ell)$ is as claimed.  Let  {\bf 0} denote an all-zero sequence and $n=1$. It is clear that {\bf 0} is the unique sequence with $\LC_1=0$ and that there are $q-1$ sequences
$(s_1)$ of complexity 1.
Suppose inductively that the result is true for sequences of length
 $n-1\geq 1$. We consider three cases.

(a) $\ell=0,n$. Let $\ell=0$. Then clearly $N(n,\ell)\geq 1$. If $\LC_n(s)=0$ then $s^{(n-1)}={\bf 0}$ by the inductive
hypothesis since $0\leq \ell_{n-1} \leq \LC_n=\ell$ and so $N(n,\ell)=1$. Suppose now that $\ell=n$. We show that $N(n,\ell)=q-1$.
If $s^{(n-1)}={\bf 0}$ and $\Delta_n=s_{n} \neq 0$ 
then $\LC_n=n$, so $N(n,\ell)\geq q-1$. Moreover, $\LC_{n-1}\leq n-1$ and
$n=\LC_n=\max\{\LC_{n-1},n-\LC_{n-1}\}$
forces $\LC_{n-1}=0$, so $s^{(n-1)}={\bf 0}$  and thus $N(n,n)=q-1$.

(b) $1\leq \ell \leq \lfloor n/2 \rfloor$. 
Suppose first that $2\ell \leq n-1$. Then $\LC_{n-1}\leq \LC_n=\ell \leq 
\lfloor (n-1)/2 \rfloor$ and we can apply the
inductive hypothesis to any $s^{(n-1)}$. If $s_{n}$ is such that 
$\Delta_n=0$ for some $s^{(n-1)}$, then
$1 \leq \ell=\LC_n=\LC_{n-1} \leq \lfloor (n-1)/2 \rfloor$, and we obtain
$N(n-1,\ell)=q^{2\ell-1}(q-1)$ sequences in this way. We also have $\ell < n-\ell$,
so $\ell$ cannot result from some $s^{(n-1)}$ with $\Delta_n\neq 0$.
Thus  $N(n,\ell)=N(n-1,\ell)=q^{2\ell-1}(q-1)$ as required.

Suppose now that $2\ell=n$. Then $\ell> \lfloor (n-1)/2  \rfloor$.
If $s_{n}$ is such that $\LC_n=\ell$ and $\Delta_n=0$, the inductive
hypothesis yields $N(n-1,\ell)=q^{2(n-1-\ell)}(q-1)$ sequences. There are also
$(q-1)N(n-1,\ell)=q^{2(n-1-\ell)}(q-1)^2$
sequences resulting from $\Delta_n\neq 0$. Thus
$$N(n,\ell)=N(n-1,\ell)+(q-1)N(n-1,\ell),$$
and substituting the inductive values and $n=2\ell$ yields the result.

(c) $\lfloor n/2 \rfloor < \ell \leq n$. Then $(n-1)/2 < \ell$
and $\max\{\ell,n-\ell\}=\ell$.
 If $s_{n}$ is such that $\Delta_n=0$, then
$\lfloor (n-1)/2 \rfloor < \ell=\LC_{n-1} \leq n-1$
and we can apply the inductive hypothesis to $s^{(n-1)}$, giving
$N(n-1,\ell)=q^{2(n-1-\ell)}(q-1)$ sequences. We also get a sequence of complexity
$\ell$ if $\Delta_{n}\neq 0$ and either (i) $\LC_{n-1}=\ell$ or
(ii) $\LC_{n-1}=n-\ell$.
Since    $\lfloor (n-1)/2 \rfloor < \ell=\LC_{n-1} \leq n-1$,
(i) gives $(q-1)N(n-1,\ell)=q^{2(n-1-\ell)}(q-1)^2$ sequences.
For (ii), we have $1 \leq n-\ell \leq  \lfloor (n-1)/2 \rfloor$
and so we obtain an additional 
$(q-1)N(n-1,n-\ell)=q^{2(n-\ell)-1}(q-1)^2$ sequences. Thus
$$N(n,\ell)=N(n-1,\ell)+(q-1)N(n-1,\ell)+(q-1)N(n-1,n-\ell)$$
and on substituting the inductive values, we easily get $N(n,\ell)=q^{2(n-\ell)}(q-1)$ as required.
\end{proof}

\begin{center}{\bf Corrigenda}
\end{center}
We take this opportunity to correct some typographical errors in \cite{N99b}:

p. 335, l. 6. delete $\varepsilon(g)+\deg\ g\leq m$.

p. 336, l. 2.  should read $\mathcal{O}(X^2-X)=((X^2-X)\circ\mathcal{F}')_{-3+2}=\mathcal{F}'_{-3}-\mathcal{F}'_{-2}=1$.
p. 336 line 13 $n<m$ should be $m\leq n$.

p. 343, table for 1,1,2 iterations: $\mathcal{O}\mu_{-1}=-1$,
$\mathcal{O}\mu_{-2}=+1$.


\begin{thebibliography}{10}

\bibitem{AS}
A.~Alecu and A.~Salagean.
\newblock {Modified Berlekamp-Massey Algorithm for Approximating the $k$-Error
  Linear Complexity of Binary Sequences}.
\newblock {\em {I.M.A. Conference on Cryptography and Coding (S.D. Galbraith,
  Ed.)}: Springer LNCS vol. 4887}, pages 220--232, 2007.

\bibitem{ABN}
F.~Arnault, Berger T.P., and A.~Necer.
\newblock {Feedback with Carry Shift Registers Synthesis With the Euclidean
  Algorithm.}
\newblock {\em {IEEE Trans. on Information Theory}}, 50:910--916, 2004.

\bibitem{BL83}
R.~Blahut.
\newblock {\em Theory and Practice of Error Control Codes}.
\newblock Addison-Wesley, 1983.

\bibitem{FT91}
G.~L. Feng and K.~K. Tzeng.
\newblock A generalization of the {Berlekamp-Massey} algorithm for
  multisequence shift register sequence synthesis with applications to decoding
  cyclic codes.
\newblock {\em IEEE Trans. Inform. Theory}, 37:1274--1287, 1991.

\bibitem{FJ98}
P.~Fitzpatrick and S.~Jennings.
\newblock Comparison of two algorithms for decoding alternant codes.
\newblock {\em Applicable Algebra in Engineering, Communications and
  Computing}, 9:211--220, 1998.

\bibitem{FN95}
P.~Fitzpatrick and G.H. Norton.
\newblock The {Berlekamp-Massey} algorithm and linear recurring sequences over
  a factorial domain.
\newblock {\em Applicable Algebra in Engineering, Communication and Computing},
  6:309--323, 1995.

\bibitem{Gus76}
F.G. Gustavson.
\newblock Analysis of the {Berlekamp-Massey} linear feedback shift-register
  synthesis algorithm.
\newblock {\em IBM J. Res. Dev.}, 20:204--212, 1976.

\bibitem{HJ}
A.~E. Heydtmann and J.M.. Jensen.
\newblock {On the Equivalence of the Berlekamp-Massey and the Euclidean
  Algorithms for Decoding.}
\newblock {\em {IEEE Trans. on Information Theory}}, 46:2614--2624, 2000.

\bibitem{IY}
K.~Imamura and W.~Yoshida.
\newblock {A Simple Derivation of the Berlekamp-Massey Algorithm and Some
  Applications.}
\newblock {\em {IEEE Trans. on Information Theory}}, 33:146---150, 1987.

\bibitem{Wirth}
K.~Jensen and N.~Wirth.
\newblock {\em {Pascal: User Manual and Report (2nd Edition)}}.
\newblock Springer, 1978.

\bibitem{JM}
E.~Jonckheere and C.~Ma.
\newblock {A Simple Hankel Interpretation of the Berlekamp-Massey Algorithm.}
\newblock {\em {Linear Algebra and its Applications}}, 125:65---76, 1989.

\bibitem{LN83}
R.~Lidl and H.~Niederreiter.
\newblock {\em Finite Fields, Encyclopedia of Mathematics and its
  Applications}, volume~20.
\newblock Addison-Wesley, Reading, 1983.

\bibitem{Ma69}
J.~L. Massey.
\newblock Shift-register synthesis and {BCH} decoding.
\newblock {\em IEEE Trans. Inform. Theory}, 15:122--127, 1969.

\bibitem{McE02}
R.~McEliece.
\newblock {\em The Theory of Information and Coding (Encyclopedia of
  Mathematics and its Applications)}, volume~3.
\newblock Cambridge University Press, 2002.

\bibitem{Nied90}
H.~Niederreiter.
\newblock The linear complexity profile and the jump complexity of keystream
  sequences.
\newblock {\em {Lecture Notes in Computer Science}}, 473:174--188, 1990.

\bibitem{N09d}
G.~H. Norton.
\newblock {Minimal Polynomial Algorithms for Finite Sequences.}
\newblock {\em {IEEE Trans. on Information Theory}}, 56:4643--4645, 2010.

\bibitem{N10b}
G.~H. Norton.
\newblock {On Minimal Polynomial Identities for Finite Sequences.}
\newblock {\em {Submitted}}, pages 1--25, 2010.

\bibitem{N95b}
G.H. Norton.
\newblock {On the Minimal Realizations of a Finite Sequence}.
\newblock {\em J. Symbolic Computation}, 20:93--115, 1995.

\bibitem{N95c}
G.H. Norton.
\newblock Some decoding applications of minimal realization.
\newblock In {\em Cryptography and Coding}, volume 1025, pages 53--62. Lecture
  Notes in Computer Science. Springer, 1995.

\bibitem{N99b}
G.H. Norton.
\newblock On shortest linear recurrences.
\newblock {\em J. Symbolic Computation}, 27:323--347, 1999.

\bibitem{NS-key}
G. H. Norton and A.~Salagean.
\newblock On the key equation over a commutative ring.
\newblock {\em Designs, Codes and Cryptography}, 20:125--141, 2000.

\bibitem{PW}
W.~W. Peterson and W.J. Weldon.
\newblock {\em Error-Correcting Codes}.
\newblock MIT Press, 1972.

\bibitem{divfree}
I.S. Reed, M.T. Shih, and T.K. Truong.
\newblock {VLSI design of inverse-free Berlekamp-Massey algorithm}.
\newblock {\em {IEE Proc. E, Computers and Digital Techniques}},
  138:5:295--298, 1991.

\bibitem{R86}
R.A. Rueppel.
\newblock {\em Analysis and Design of Stream Ciphers.}
\newblock Springer, 1986.

\bibitem{S05}
A.~Salagean.
\newblock {On the Computation of the Linear Complexity and the $k$-Error Linear
  Complexity of Binary Sequences With Period a Power of 2}.
\newblock {\em {IEEE Trans. Inform. Theory}}, 51:1145--1150, 2005.

\bibitem{Salagean}
A.~Salagean.
\newblock {An Algorithm for Computing Minimal Bidirectional Linear Recurrence
  Relations.}
\newblock {\em {IEEE Trans. Info. Theory}}, 55:4695--4700, 2009.

\end{thebibliography}
\end{document}